\documentclass[twocolumn,floatfix,aps,superscriptaddress,showpacs,amsfonts]{revtex4}
\usepackage{graphicx}
\usepackage{array}
\usepackage{amsthm}
\usepackage{amsmath}
\usepackage{amssymb}

\theoremstyle{plain}
\newtheorem*{thm*}{Theorem}

\begin{document}

\title{Optimal measurements to access classical correlations of
  two-qubit states}

\author{Xiao-Ming Lu} \affiliation{Zhejiang Institute of Modern
  Physics, Department of Physics, Zhejiang University, Hangzhou
  310027, China}

\author{Jian Ma} \affiliation{Zhejiang Institute of Modern Physics,
  Department of Physics, Zhejiang University, Hangzhou 310027, China}

\author{Zhengjun Xi} \affiliation{Zhejiang Institute of Modern
  Physics, Department of Physics, Zhejiang University, Hangzhou
  310027, China} \affiliation{College of Computer Science, Shaanxi
  Normal University, Xi'an 710062, China}

\author{Xiaoguang Wang} \affiliation{Zhejiang Institute of Modern
  Physics, Department of Physics, Zhejiang University, Hangzhou
  310027, China}
\begin{abstract}
  We analyze the optimal measurements accessing classical correlations
  in arbitrary two-qubit states. Two-qubit states can be transformed
  into the canonical forms via local unitary operations. For the
  canonical forms, we investigate the probability distribution of the
  optimal measurements.  The probability distribution of the optimal
  measurement is found to be centralized in the vicinity of a specific
  von Neumann measurement, which we call the
  maximal-correlation-direction measurement (MCDM). We prove that for
  the states with zero-discord and maximally mixed marginals, the MCDM
  is the very optimal measurement.  Furthermore, we give an upper
  bound of quantum discord based on the MCDM, and investigate its
  performance for approximating the quantum discord.
\end{abstract}
\pacs{03.67.-a,03.65.Ta}
\maketitle

\section{introduction}

Optimization procedures are involved in many quantities in the quantum
information theory.  One of the most important examples is the
entanglement of formation, which is defined as the least expected
entanglement over all ensembles of pure states realizing the given
states~\cite{Bennett1996}.  Here we focus on the optimization
procedure involved in the classical correlation, which is defined as
the maximal amount of information about the subsystem $B$ that can be
obtained via performing measurement on the other subsystem
$A$~\cite{Henderson2001}. Meanwhile, the non-classical correlation is
measured by the so-called quantum discord~\cite{Ollivier2001}, which
is the difference of the total amount of correlations and the
classical correlation. Over the past decade, quantum discord has
received a lot of attention, including analytic calculations for some
special two-qubit sates~\cite{Luo2008,Ali2010}, detecting of quantum
discord~\cite{Bylicka2010,Dakic2010,Datta2009a,Datta2010,Ferraro2010,Luo2009,Maziero2010,Modi2010},
quantum discord in the concrete physical
models~\cite{Bradler2009,Chen2010a,Cui2009,Datta2009,Fanchini2010a,Maziero2010a,Sarandy2009,Soares-Pinto2010,Wang2010a,Werlang2010},
the dynamics of quantum
discord~\cite{Chakrabarty2010,Cole2008,Fanchini2010,Maziero2009,Mazzola2010,Mazzola2010a,Lu2010,Maziero2010b,Vasile2010,Wang2010,Werlang2009,Yuan2010},
quantum discord in continuous variable
system~\cite{Adesso2010,Giorda2010}, quantum discord of multipartite
states~\cite{Chen2010}, and exploration in the
laboratory~\cite{Lanyon2008,Xu2010a,Xu2010}, etc.  Most recently,
operational interpretations of quantum discord are proposed in
Ref.~\cite{Cavalcanti2010,Madhok2010}, where quantum discord was shown
to be a quantitative measure about the performance in the quantum
state merging.

To access the total classical correlation or get the value of quantum
discord, one has to find the corresponding optimal measurements. On
the other hand, the sudden change of the optimal measurements during
evolution is related to a novel phenomenon of the quantum
discord---the sudden change in decay
rates~\cite{Maziero2009,Xu2010a,Mazzola2010,Mazzola2010a,Lu2010}. So a
study on the optimal measurements will help us to understand the
dynamical properties of the quantum discord. The optimization involved
is taken over general measurements, which is described by a
positive-operator-valued measure (POVM). In Ref.~\cite{Hamieh2004},
Hamieh \emph{et al.} showed that for two-qubit system the optimal
measurement over POVM is a projective measurement. In this paper, we
only consider orthogonal projective measurements , knows as von
Neumann measurement. Despite recent progress, the optimization
procedure over von Neumann measurements is still hard to be resolved
for general two-qubit states, and analytic approaches still lack. This
motives us to systematically investigate the optimal measurements and
find an effective way to approach the optimal measurement.

In the present article, we investigate the probability distribution of
the optimal measurements. Based on a general analysis on those factors
which influence the classical correlation, we introduce the canonical
forms for the two-qubit states and the maximal-correlation-direction
measurement (MCDM). For arbitrary two-qubit states, the optimal
measurements are found to be centralized in the vicinity of the MCDM.
We prove that for the states with zero-discord and maximally mixed
marginal, the MCDM is the very optimal measurement accessing classical
correlation.  We also study its veracity for $X$-states and arbitrary
states. Then we propose the MCDM-based discord as an upper bound of
quantum discord. It is demonstrated that the MCDM-based could be a
good approximation of quantum discord.

This article is organized as follows. In Sec.~\ref{sec:review-of-QD},
we give a brief review on the measures of the total, quantum and
classical correlations. In Sec.~\ref{sec:MCDM}, first we give a
general analysis on those factors which influence the classical
correlations.  Then we introduced the concepts of the canonical form
and the MCDM. In Sec.~\ref{sec:performance-of-MCDM}, we investigate
the probability distribution of the optimal measurements and the
performance of the MCDM. In Sec.~\ref{sec:MCDM-based-discord}, we
propose the MCDM-based discord as an upper bound of quantum discord,
and investigate its performance for approximating quantum
discord. Section \ref{sec:conclusion} is the conclusion.

\section{review of quantum discord\label{sec:review-of-QD}}

First, we recall the concepts of the total amount of correlations, the
classical correlation and the quantum correlation. Given a quantum
state $\rho$ in a composite Hilbert space
$\mathcal{H}=\mathcal{H}^{A}\otimes\mathcal{H}^{B}$, the total amount
of correlation is quantified by the quantum mutual
information~\cite{Groisman2005}
\begin{equation}
  \mathcal{I}(\rho)=S(\rho^{A})+S(\rho^{B})-S(\rho),
  \label{eq:mutual_information}
\end{equation}
where $S(\rho)\equiv-\mathrm{Tr}\left[\rho\log_{2}\rho\right]$ is the
von Neumann entropy and $\rho^{A(B)}=\mathrm{Tr}_{B(A)}[\rho]$ is the
reduced density matrix by tracing out system $B(A)$. The total amount
of correlations can be splitted into the quantum and the classical
parts~\cite{Ollivier2001,Henderson2001}. The classical correlation is
seen as the amount of information about the subsystem $B$ that can be
obtained via performing a measurement on the other subsystem $A$
. Then the measure of the classical correlation is defined by
\begin{equation}
  \mathcal{C}(\rho)=S(\rho^{B})-\min_{\{E_{k}^{A}\}}\sum_{k}p_{k}S(\rho_{k}^{B}),
  \label{eq:classical_correlation}
\end{equation}
where $\{E_{k}^{A}\}$ is the POVM performed on $A$ and
$\rho_{k}^{B}:=\mathrm{Tr}_{A}\left(E_{k}^{A}\otimes\openone^{B}\rho\right)/p_{k}$
is the remaining state of $B$ after obtaining the outcome $k$ on $A$
with the probability
$p_{k}:=\mathrm{Tr}$$\left(E_{k}^{A}\otimes\openone^{B}\rho\right)$.
$\openone^{A(B)}$ is the identity operator on the subsystem $A(B)$.
Meanwhile, the quantum correlation is measured via the quantum discord
defined by~\cite{Ollivier2001,Henderson2001}
\begin{equation}
  \mathcal{D}(\rho)=\mathcal{I}(\rho)-\mathcal{C}(\rho),
  \label{eq:quantum_discord}
\end{equation}
which is the difference of the total amount of correlation
$\mathcal{I}(\rho)$ and the classical correlation $\mathcal{C}(\rho)$.

The definition of the classical and quantum correlations involves an
optimization process to minimize the term
$\sum_{k}p_{k}S(\rho_{k}^{B})$, which is considered as the quantum
version of the conditional entropy~\cite{Ollivier2001}. For two-qubit
systems, it is shown that the optimal POVM is a projective
measurement~\cite{Hamieh2004}.  Hereafter, we will only consider the
orthogonal measurement projective measurement, knows as von Neumann
measurement.  A von Neumann measurement $\{\Pi_1^A,\Pi_2^A\}$ of a
two-qubit system can be characterized by a unit vector
$n=(n_{1},n_{2},n_{3})^{T}$ on the Bloch sphere, through
\begin{eqnarray}
  \Pi_{1}^{A} & = & \frac{1}{2}\left(\openone^{A}+\sum_{i=1}^{3}n_{i}\sigma_{i}^{A}\right),\nonumber \\
  \Pi_{2}^{A} & = & \frac{1}{2}\left(\openone^{A}-\sum_{i=1}^{3}n_{i}\sigma_{i}^{A}\right).
  \label{eq:qubit_measurement}
\end{eqnarray}
So the optimization may be taken over half of the Bloch sphere, since
exchanging $\Pi_{1}^{A}$ and $\Pi_{2}^{A}$ gives the same measurement.

\section{\label{sec:MCDM}general analysis and
  maximal-correlation-direction measurement}

The classical correlation (\ref{eq:classical_correlation}) can be
expressed as
$\mathcal{C}=\max_{\{\Pi_{k}^{A}\}}\sum_{k}p_{k}S(\rho_{k}^{B}||\rho^{B})$,
where
$S(\rho_{k}^{B}||\rho^{B}):=-S(\rho_{k}^{B})-\mathrm{Tr}(\rho_{k}^{B}\log\rho^{B})$
is the relative entropy \cite{Henderson2001}. If we consider the
relative entropy as a measure of the distance, then the classical
correlation can be considered as maximal average distance between the
remaining state $\rho_{k}^{B}$ and the reduced state $\rho^{B}$.  It
seems that the optimal measurement tends to be the one which makes the
remaining states far away from the reduced state, although the
probability $p_{k}$ also plays an important role in the actual
optimization problem.

Based on this tendency, we use the following Fano-Bloch decomposition
of an arbitrary two-qubit state\cite{Bloch1946,Fano1983}:

\begin{equation}
  \rho=\rho^{A}\otimes\rho^{B}+\frac{1}{4}\sum_{ij=1}^{3}\Lambda_{ij}\sigma_{i}^{A}\otimes\sigma_{j}^{B},
  \label{eq:correlation_matrix}
\end{equation}
where
$\Lambda_{ij}=\langle\sigma_{i}^{A}\sigma_{j}^{B}\rangle_{\rho}-\langle\sigma_{i}^{A}\rangle_{\rho}\langle\sigma_{j}^{B}\rangle_{\rho}$
is the correlation function with $\langle
O\rangle_{\rho}:=\mathrm{Tr}\left[\rho O\right]$ defined. This form
was used to investigate the dynamics of open quantum systems in the
presence of initial correlation~\cite{vStelmachovivc2001}, and the
correlation functions were used to characterized the correlations in a
quantum state of a composite system~\cite{Huang2008,Dong2010}.
Because the quantum and classical correlations are both invariant
under local unitary transformations, we will consider a special set of
two-qubit states, in which for arbitrary two-qubit state we can find
an equivalent state up to local unitary transformations. In
Ref.~\cite{Luo2008}, Luo showed that the matrix $\Lambda$ in
Eq.~(\ref{eq:correlation_matrix}) can be diagonalized through local
unitary transformations. Furthermore, the order and the signs of the
eigenvalues of $\Lambda$ can be realigned through $SO(3)$
transformations, which correspond to $SU(2)$ transformations on the
density matrices. Here we introduce the canonical form of two-qubit
states, which is defined by

\begin{equation}
  \rho=\rho^{A}\otimes\rho^{B}+\frac{1}{4}\sum_{i=1}^{3}\Lambda_{i}\sigma_{i}^{A}\otimes\sigma_{i}^{B}
  \label{eq:canonical_form}
\end{equation}
with $\Lambda_{i}=\langle\sigma_{i}^{A}\sigma_{i}^{B}\rangle_{\rho}-\langle\sigma_{i}^{A}\rangle_{\rho}\langle\sigma_{i}^{B}\rangle_{\rho}$
and $\Lambda_{1}\geq\Lambda_{2}\geq|\Lambda_{3}|$. The sign of
$\Lambda_3$ is determined by the determinant $|\Lambda|$ via $|\Lambda
|=\Lambda_1\Lambda_2\Lambda_3$.  An arbitrary two-qubit state is
equivalent to a certain state in the canonical form up to local
unitary transformations. So studying the states in the canonical form
is adequate to understand the quantum and classical correlations of
arbitrary two-qubit states. Hence, we will only consider the states in
the canonical form hereafter.

After von Neumann measurement characterized by
Eq.~(\ref{eq:qubit_measurement}) performed on $A$, we get the
remaining state of $B$ with the outcome $k=1,2$ on $A$ as follows
\begin{equation}
  \rho_{k}^{B}=\rho^{B}+\frac{1}{p_{k}}\Delta_{k}
  \label{eq:decomposition_remaining_state}
\end{equation}
where $p_{k}=\mathrm{Tr}\left[\rho^{A}\Pi_{k}^{A}\right]$ and
$\Delta_{k}\equiv\frac{1}{4}(-1)^{k+1}\sum_{i=1}^{3}n_{i}\Lambda_{i}\sigma_{i}^{B}$
are both dependent on $\Pi_{k}^{A}(n_{i})$. Note that
$\Delta_{2}=-\Delta_{1}$.  From
Eq.~(\ref{eq:decomposition_remaining_state}), we can see that the von
Neumann measurement influences the quantum conditional entropy through
$p_{k}$ and $\Delta_{k}$. The influence of $p_{k}$ is complicated,
since $p_{k}$ not only impacts the remaining state $\rho_{k}^{B}$ but
also impacts the averaging of $S(\rho_{k}^{B})$, while $\Delta_{k}$
only impacts the remaining state. If we assume that the influence of
$\Delta_{k}$ to the quantum conditional entropy is stronger than that
of $p_{k}$, the optimal measurement will tend to be the one which
maximizes
$|\Delta_{1}|\equiv\sqrt{\mathrm{Tr}\Delta_{1}\Delta_{1}^{\dagger}}$
(Remind $|\Delta_{2}|=|\Delta_{1}|$). After some algebras, we have
\begin{equation}
  |\Delta_{1}|^{2}=\frac{1}{16}\sum_{i=1}^{3}\Lambda_{i}^{2}n_{i}^{2}\leq\frac{1}{16}\Lambda_{1}^{2}.
  \label{eq:MCDM}
\end{equation}
The maximum of $|\Delta_{1}|^{2}$ is achieved when $n=(1,0,0)^{T}$.
So the corresponding measurement is
\begin{equation}
  \left\{\Pi_{1}^{A}=\frac{1}{2}(\openone+\sigma_{1}^{A}),\,\Pi_{2}^{A}=\frac{1}{2}(\openone-\sigma_{1}^{A})\right\}.
  \label{eq:MCDM_canonical}
\end{equation}
Because $\Lambda_{i}$ are correlation functions and $\Lambda_{1}$ is
the largest one of them, we call Eq.~(\ref{eq:MCDM_canonical}) the
maximal-correlation-direction measurement (MCDM) performed on $A$ for
the two-qubit states in the canonical form. For an arbitrary two-qubit
state, we can always find a local unitary transformation $U_{1}\otimes
U_{2}$ such that $\tilde{\rho}=U_{1}\otimes U_{2}\rho
U_{1}^{\dagger}\otimes U_{2}^{\dagger}$ is in the canonical form. So
the MCDM performed on $A$ for an arbitrary two-qubit state is given by
\begin{equation}
  \left\{ \frac{1}{2}\left(\openone+U_{1}^{\dagger}\sigma_{1}^{A}U_{1}\right),\frac{1}{2}\left(\openone-U_{1}^{\dagger}\sigma_{1}^{A}U_{1}\right)\right\} .
  \label{eq:MCDM_arbitrary}
\end{equation}

\section{\label{sec:performance-of-MCDM}performance of the MCDM}
In the following, we will investigate the performance of the MCDM to
access the classical correlation measured by the quantum discord.

\subsection{Zero-discord states}
The zero-discord states are the states satisfying
$\rho=\sum_{k}\Pi_{k}^{A}\otimes\openone^{B}\rho\Pi_{k}^{A}\otimes\openone^{B}$,
where $\{\Pi_{k}^{A}\}$ is just the optimal von Neumann measurement to access the
classical correlation, see Ref.~\cite{Ollivier2001}. In Appendix
\ref{sec:appd_zero-discord}, we show that a sufficient and necessary
condition of zero-discord is the existence of such a unit vector $n$
satisfying the following equations
\begin{eqnarray}
  nn^{T}a & = & a,\label{eq:zero-QD-1}\\
  nn^{T}R & = & R,\label{eq:zero-QD-2}
\end{eqnarray}
where $a_{i}=\mathrm{Tr}(\sigma_{i}^{A}\rho^{A})$ is the polarization
vector of the reduced density matrix $\rho^{A}$ , $R$ is a $3\times3$
matrix with the elements
$R_{ij}=\mathrm{Tr}(\sigma_{i}^{A}\otimes\sigma_{j}^{B}\rho)$, and $n$
is the column vector characterizing the optimal von Neumann
measurement via Eq.~(\ref{eq:qubit_measurement}). For the states in
the canonical form, we have $R=\Lambda+ab^{T}$, see the definitions of
$R$, $\Lambda$, $a$ and $b$. Then the conditions (\ref{eq:zero-QD-1})
and (\ref{eq:zero-QD-2}) leads to
\begin{equation}
  \Lambda=nn^{T}\Lambda.
  \label{eq:zero-discord-condition}
\end{equation}
Considering
$\Lambda=\mathrm{diag\{\Lambda_{1},\Lambda_{2},\Lambda_{3}\}}$ is a
diagonal matrix with $\Lambda_{1}\geq\Lambda_{2}\geq|\Lambda_{3}|$,
from Eq.~(\ref{eq:zero-discord-condition}), the first diagonal element
of the matrix $\Lambda$ reads
\begin{equation}
  \Lambda_{1}=\left(n_{1}\right)^{2}\Lambda_{1}.
\end{equation}
So we obtain $n_{1}=\pm1$, which corresponds to the MCDM.

So we conclude that for the zero-discord states, the MCDM is just the
optimal measurement to access the classical correlation.

\subsection{states with maximally mixed marginals\label{MMM_states}}
The states with maximally mixed marginals are the ones satisfying
$\rho^{A}=\openone^{A}/d_{A}$ and $\rho^{B}=\openone^{B}/d_{B}$, where
$d_{A(B)}=\mathrm{dim}(\mathcal{H}^{A(B)})$ is the dimension of the
Hilbert-space $\mathcal{H}^{A(B)}$. For two-qubit systems, The
canonical forms of states with maximally mixed marginals must be
Bell-diagonal states, while an arbitrary Bell-diagonal state need not
be in the canonical form. An arbitrary Bell-diagonal state reads
\begin{equation}
  \rho=\frac{1}{4}\openone^A\otimes\openone^B+\frac{1}{4}\sum_{i=1}^3
  c_i\sigma_i^A \otimes\sigma^B_i
\end{equation}
with $c_i=\langle\sigma_i^A\otimes\sigma_i^B\rangle$. Only when
$c_i\geq c_2\geq |c_3|$, the Bell-diagonal state is in the canonical
form. For Bell-diagonal states, Luo got the analytical results of
quantum discord~\cite{Luo2008}. There it was shown that the optimal
measurement is given by the unit vector $n=(n_1,n_2,n_3)^T$ which
maximizes $\sqrt{\sum_{i=1}^{3}c_{i}^{2}n_{i}^{2}}$. So the optimal
measurements of Bell-diagonal states must be in the \emph{universal}
finite set $\{\{\frac{1}{2}(\openone\pm \sigma_i^A)
\}|i=1,2,3\}$. Here \emph{universal} means this set of von Neumann
measurements is independent on the given states.  For the
Bell-diagonal states in the canonical form, we have $c_{1}\geq c_{2}
\geq|c_{3}|$.  Then the optimal measurement is explicitly given by
$n=(1,0,0)^{T}$, which is consistent with the MCDM
(\ref{eq:MCDM_canonical}). This is because $\Lambda_{i}=c_{i}$ due to
$\langle\sigma_{i}^{A}\rangle_{\rho}=0$ (maximal mixed margins), then
the maximization of $\sqrt{\sum_{i=1}^{3}c_{i}^{2}n_{i}^{2}}$ is
equivalent to the maximization of $|\Delta_{1}|^{2}$, see
Eq.~(\ref{eq:MCDM}).

So we conclude that for the states with maximally mixed marginals, the
MCDM is just the optimal measurement to access the classical
correlation.

\subsection{$X$-states}

In the following, we consider the so-called $X$-states, named because
of the visual appearance of the density matrix
\begin{equation}
  \rho=\left[\begin{array}{cccc}
      a & 0 & 0 & w^{*}\\
      0 & b & z^{*} & 0\\
      0 & z & c & 0\\
      w & 0 & 0 & d
    \end{array}\right].
\end{equation}
The $X$-states, including maximally entangled Bell states and Werner
states, are a class of typical quantum states in the field of quantum
information. An algebraic characterization of $X$-states is presented
in Ref.~\cite{Rau2009}. The Fano-Bloch representation matrix
$\tau_{ij}=\mathrm{Tr}(\sigma_{i}^{A}\otimes\sigma_{j}^{B}\rho)$ of
$X$-states is also of $X$-type
\begin{equation}
  \tau=\left[\begin{array}{cccc}
      1 & 0 & 0 & b_{3}\\
      0 & R_{11} & R_{12} & 0\\
      0 & R_{21} & R_{22} & 0\\
      a_{3} & 0 & 0 & R_{33}\end{array}\right].
\end{equation}
An important property of the class of the $X$-states is that an
$X$-state after the local unitary operations
$\exp(i\sigma^A_3\varphi_1/2)\otimes\exp(i\sigma^B_3\varphi_2/2)$ is
also an $X$-state. This leads to the following theorem:
\begin{thm*}
  For the $X$-states, it is impossible to exist a \textbf{universal
    finite} set of von Neumann measurements among which the optimal
  measurement must be.
\end{thm*}
\begin{proof}
  The proof is divided into two steps. First, we assume that for the
  class of $X$-state, such a universal finite set exists and is
  denoted by $\mathcal{M}:=\{\mathcal M_i|i\in \mathcal{I}\}$, where
  $\mathcal{M}_i$ is a von Neumann measurement and $\mathcal{I}$ is a
  finite set of the indexes.  Let $R_z(\varphi_1,\varphi_2):=
  R_z^A(\varphi_1)\otimes R_z^B(\varphi_2)$ with
  $R_z^{A(B)}(\varphi):=\exp(i\sigma^{A(B)}_3\varphi/2)$. If
  $\mathcal{M}_i$ is the optimal measurement for a given state $\rho$,
  then
  $\tilde{\mathcal{M}}_i(\varphi_1)=R_z(\varphi_1)\mathcal{M}_iR^\dagger_z(\varphi_1)$
  is the optimal measurement for the state
  $\tilde{\rho}(\varphi_1,\varphi_2)=R_z(\varphi_1,\varphi_2)\rho
  R_z^\dagger(\varphi_1,\varphi_2)$. Because
  $\tilde{\rho}(\varphi_1,\varphi_2)$ is also an $X$-state, we must
  have $\tilde{\mathcal{M}}_i(\varphi_1)\in\mathcal{M}$. Remind that
  $\mathcal{M}$ is a finite set, the only possibility is that all the
  $\mathcal{M}_i$ are invariant under the operation $R_z(\varphi_1)$
  with arbitrary $\varphi_1$, i.e., the only possibility of
  $\mathcal{M}_i$ is $\{(\openone\pm\sigma_3^A)/2\}$. In the second
  step, we disprove this only possibility. From Sec.~\ref{MMM_states},
  we already know that for Bell-diagonal states, there are three
  possibility of the optimal measurement: $\{(\openone\pm\sigma_i^A)/2\}$
  for $i=1,2,3$. These three measurements constitute the minimal
  universal finite set $\mathcal{M}_\mathrm{BD}$ of possible optimal
  measurement for Bell-diagonal states. Because Bell-diagonal states
  are a subclass of the $X$-states, $\mathcal{M}_\mathrm{BD}$ must be
  a subset of $\mathcal{M}$, but actually it is not. This disproves
  the only possibility derived in the first step. So the above theorem
  is proved.
\end{proof}
The above theorem implies for the entire class of $X$-states, the
optimization procedure involved in the classical correlation must be
state-dependent. This result is opposite to that of
Ref.~\cite{Ali2010}, where the authors gave a universal finite set of
candidates and sought the optimal measurement in this set. In
Ref.~\footnote{ This conclusion is opposite with the results of
  Ref.~\cite{Ali2010}. In Ref.~\cite{Ali2010}, the optimization is
  taken with respect to the four parameters $m$, $n$, $k$ and $l$ with
  one constraint $k+l=1$, see Eq.~(18) in
  Ref.~\cite{Ali2010}. However, only two independent parameters is
  needed to characterize a von Neumann measurement on two-qubit
  systems, so there should be one more constraint. The four parameters
  $m$, $n$, $k$, $l$ are related to three other parameter $z_1$, $z_2$
  and $z_3$ in Ref.~\cite{Luo2008} through $4m=z_2^2$, $4n=-z_1 z_2$,
  $k-l=z_3$. The three parameters $z_1$, $z_2$, $z_3$ satisfy
  $z_1^2+z_2^2+z_3^2=1$~\cite{Luo2008}, which implicitly gives another
  constraint on the four parameters $m$, $n$, $k$ and $l$. This
  implicit constraint was not taken into consideration in the
  optimization procedure in Ref.~\cite{Ali2010}. So the analytical
  result of quantum discord for $X$-states has not been obtained up to
  now.  }, we show the constraint missed in Ref.~\cite{Ali2010}. To
elucidate this, we also give an explicit example:
\begin{equation}
  \rho=\left[
    \begin{array}{cccc}
      0.0783 & 0 & 0 & 0\\
      0 & 0.1250 & 0.1000 & 0\\
      0 & 0.1000 & 0.1250 & 0\\
      0 & 0 & 0 & 0.6717
    \end{array}
  \right].
  \label{eq:anti_example}
\end{equation}
This state is in the canonical form with $\Lambda_{1}=\Lambda_{2}=0.2$
and $\Lambda_{3}=0.1479$. The optimal measurements are characterized
by two angle $\theta$ and $\phi$, via Eq.~(\ref{eq:qubit_measurement})
and
\begin{equation}
  n=(\sin\theta\cos\phi,\sin\theta\sin\phi,\cos\theta)^T
\end{equation}
with $\theta\in [0,\pi)$ and $\phi\in[-\pi/2,-\pi/2).$ It can be
directly verified that the state (\ref{eq:anti_example}) is invariant
under the operation
$\exp(i\varphi_1\sigma_3^A)\otimes\exp(i\varphi_2\sigma_3^B)$, so the
optimization is only relevant to $\theta$, see
Ref.~\cite{Lu2010}. Following the optimization strategy of
Ref.~\cite{Ali2010}, the value of $\theta$ can only be either $0$ or
$\pi/2$. However, the optimal measurement of this state is numerically
found at $\theta\simeq0.155\pi$.

In the following, we numerically investigate the probability
distribution of the optimal measurement for the class of $X$-states.
We first generate 100,000 random density matrices according to
Hilbert-Schmidt measure~\cite{Zyczkowski2001} and then project them
into the X-state subspace via $\rho_X = \sum_{k=1,2}E_k\rho
E_k^\dagger$ with $E_1=\mathrm{diag\{1,0,0,1\}}$ and
$E_2=\mathrm{diag\{0,1,1,0\}}$. These random X-states was later
transformed into the canonical forms via local unitary
transformations. We numerically find optimal measurements to minimize
the quantum conditional entropy, utilizing the general expression of
the quantum conditional entropy, see Eq.~(\ref{eq:general_QCE}) in
Appendix \ref{sec:appd-general-expression-QCE}.  If there are more
than one optimal measurement, we chose the one closest to the MCDM,
because we concern on how to find an optimal measurement but not all
of the optimal measurements. Our numerical results in
Table~\ref{tab:X_state} show the probability that the MCDM would be
the optimal measurement is about $99.40\%$, and the second preference
of the optimal measurement is given by $n=(0,1,0)^{T}$.
\begin{table}
  \begin{centering}
    \begin{tabular}{|>{\centering}p{80pt}|>{\centering}p{100pt}|c|}
      \hline
      $\theta$ & $\phi$ & Percentage\tabularnewline
      \hline
      \hline
      $\pi/2$ & 0 & 99.40\%\tabularnewline
      \hline
      $\pi/2$ & $-\pi/2$ & 0.60\%\tabularnewline
      \hline
    \end{tabular}
    \par\end{centering}
  \caption{Distribution of the optimal measurements to access the classical correlation,
    for random density matrices in the canonical form and equivalent to
    $X$-states up to local unitary transformations. The total number
    of the random states is 10000.\label{tab:X_state}}
\end{table}

In Table~\ref{tab:X_state}, it seems that the value of $\theta$ for
the optimal measurement is always $\pi/2$, however, a counterexample
does exist and is already given in the Eq.~(\ref{eq:anti_example}).

\subsection{arbitrary two-qubit states}

For arbitrary two-qubit states, the optimization procedure involved in
the quantum discord is unreachable up to now. The MCDM solves this
optimization for the states with zero-discord and with maximally mixed
marginals. Besides, the MCDM hits the optimal measurements for the
most of the $X$-states. So what about the performance of the MCDM for
arbitrary two-qubit states?

\begin{figure}[h]
  \begin{centering}
    \includegraphics[scale=0.5]{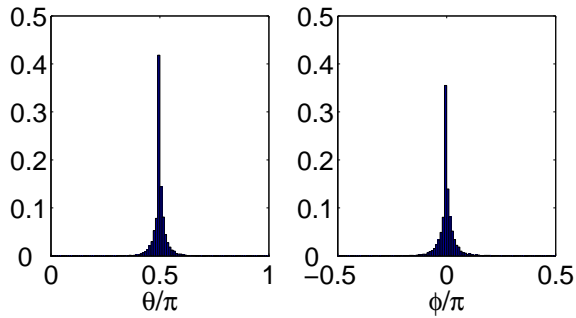}
    \par\end{centering}
  \caption{\label{fig:arbitrary_states}Probability distribution of the
    optimal measurements for the 100,000 random states according to Hilbert-Schmidt measure~\cite{Zyczkowski2001},in the canonical form. Here we discretized the domain of $\theta$ and
    $\phi$ into 100 sections.}
\end{figure}

We generate 100,000 random density matrices of two qubit according to
Hilbert-Schmidt measure~\cite{Zyczkowski2001} and transform them into
the canonical form. Then we numerically investigate the probability
distribution of the optimal measurement characterized by two angle
$\theta$ and $\phi$ via Eq.~(\ref{eq:qubit_measurement}).  In
Fig.~\ref{fig:arbitrary_states}, we show that the MCDM is indeed the
preference of the optimal measurements. Meanwhile, the optimal
measurements are centralized in the vicinity of the MCDM. This
motivates us to consider the MCDM as an alternative of the optimal
measurement accessing the classical correlation.

\section{\label{sec:MCDM-based-discord}MCDM-based discord}

The centralization of the optimal measurements in the vicinity of the
MCDM motivates us to introduce a MCDM-based discord for an arbitrary
two-qubit state as follows
\begin{equation}
  \tilde{\mathcal{D}}(\rho):=S(\rho^{A})-S(\rho)+S_{\{\tilde{\Pi}_{k}^{A}\}}(B|A),
\end{equation}
where $S_{\{\tilde{\Pi}_{k}^{A}\}}(B|A)$ is the quantum conditional
entropy $\sum_{k}p_{k}S(\rho_{k}^{B})$ based on the MCDM
$\{\tilde{\Pi}_{k}^{A}\}$ given by
Eq.~(\ref{eq:MCDM_arbitrary}). Because the MCDM is in the set of von
Neumann measurement over which the optimization involved in the
quantum discord is taken, $\tilde{\mathcal{D}}(\rho)$ will be not less
than the quantum discord $\mathcal{D}(\rho)$. In other words,
$\tilde{\mathcal{D}}(\rho)$ is an upper bound of the quantum discord
$\mathcal{D}(\rho)$. If $\tilde{\mathcal{D}}(\rho)$ is close enough to
$\mathcal{D}(\rho)$, the MCDM-based discord will be a good
approximation of the quantum discord. In the following, we investigate
how close is the MCDM-based discord to the quantum discord.

\begin{figure}
  \includegraphics[scale=0.5]{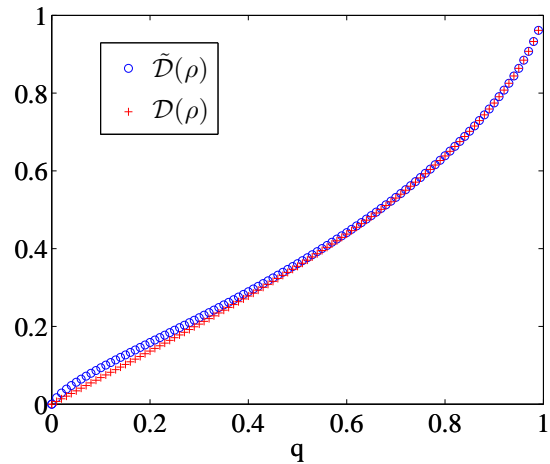}
  \caption{\label{fig:MCDM-based-discord-and-QD}MCDM-based discord and
    quantum discord for the states
    $\rho(q)=(1-q)|\psi_{0}\rangle\langle\psi_{0}|+q|\psi_{1}\rangle\langle\psi_{1}|$,
    where $|\psi_{0}\rangle$ and $|\psi_{1}\rangle$ are given by
    Eq.~(\ref{eq:mixtures}).}
\end{figure}

First, for simplicity, we consider a family of states
$\rho(q)=(1-q)|\psi_{0}\rangle\langle\psi_{0}|+q|\psi_{1}\rangle\langle\psi_{1}|$,
which are mixtures of a product state $|\psi_{0}\rangle$ and a maximal
entangled state $|\psi_{1}\rangle$, with the probability $1-q$ and $q$
respectively. Specifically, we choose $|\psi_{0}\rangle$ and
$|\psi_{1}\rangle$ as follows
\begin{eqnarray}
  |\psi_{0}\rangle & = & \frac{1}{\sqrt{2}}\left(|00\rangle+|10\rangle\right),\nonumber \\
  |\psi_{1}\rangle & = &
  \frac{1}{\sqrt{2}}\left(|01\rangle+|10\rangle\right).\label{eq:mixtures}
\end{eqnarray}
In Fig.~\ref{fig:MCDM-based-discord-and-QD}, we show that the
MCDM-based discord is very close to the quantum discord, which suggest
that it can be taken as a good approximation to the quantum discord.

\begin{figure}
  \includegraphics[scale=0.4]{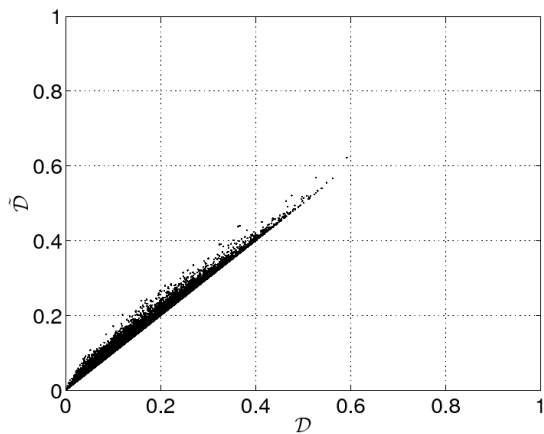}
  \caption{\label{fig:MCDM-based-discord-random}MCDM-based discord
    $\tilde{\mathcal{D}}$ versus $\mathcal{D}$, for 100,000 random
    density matrices according to the Hilbert-Schmidt
    measure~\cite{Zyczkowski2001}. The variance
    $\mathrm{avg}[(\tilde{\mathcal{D}}-\mathcal{D})^2]\simeq
    3.7433\times 10^{-5}$, where $\mathrm{avg}[\,\cdot\,]$ means the
    average taken over all the random density matrices.}

\end{figure}

For more convincing arguments, we investigate the MCDM-based discord
and the quantum discord for random states. It is shown that for the
most of the random states, the MCDM-based discord is very close to the
quantum discord, see Fig.~\ref{fig:MCDM-based-discord-random}.  This
demonstrate the efficientness of the MCDM-based discord.

\section{\label{sec:conclusion}conclusion}

In conclusion, we have investigated the probability distribution of
the optimal measurement accessing classical correlation and show that
the optimal measurements for arbitrary two-qubit state are centralized
in the vicinity of the MCDM. We have proved that the MCDM is the very
optimal measurement for the states with zero-discord and the ones with
maximally mixed marginals. Besides, we have proposed the MCDM-based
discord as an upper bound of quantum discord, and demonstrated that
the MCDM-based discord could be a good approximation of quantum
discord.
\begin{acknowledgments}
  X.~Wang is supported by National Natural Science
Foundation of China (NSFC) with Grants No. 11025527,
No. 10874151, and No. 10935010. Z.-J. Xi is supported by the Superior
  Dissertation Foundation of Shaanxi Normal University (S2009YB03).
\end{acknowledgments}

\appendix

\section{zero-discord state in Fano-Bloch
  representation\label{sec:appd_zero-discord}}

Here, we derive the condition of zero-discord states in the Fano-Bloch
representation. An arbitrary two-qubit state can be written in the
Fano-Bloch representation as follows~\cite{Schlienz1995}:

\begin{equation}
  \rho=\frac{1}{4}\sum_{ij=0}^{3}\tau_{ij}\sigma_{i}^{A}\otimes\sigma_{j}^{B}
  \label{eq:Bloch_rep}
\end{equation}
where $\sigma_{0}^{A(B)}=\openone^{A(B)}$ is the $2\times2$ identity
operator, $\sigma_{1,2,3}^{A(B)}$ is Pauli matrices. Meanwhile, a von
Neumann measurement performed on $A$ is characterized by a set of von
Neumann operators

\begin{equation}
  \Pi_{1}^{A}=\frac{1}{2}\sum_{k=0}^{3}\alpha_{k}\sigma_{k}^{A},\qquad\Pi_{2}^{A}=\frac{1}{2}\sum_{k=0}^{3}\beta_{k}\sigma_{k}^{A}.
  \label{eq:measurent_basis}
\end{equation}
The coefficients $\alpha_{k}$ and $\beta_{k}$ are given by

\begin{equation}
  \alpha_{0}=\beta_{0}=1,\,\alpha_{k}=-\beta_{k}=n_{k}\mbox{ for }k=1,2,3,
  \label{eq:alpha_beta}
\end{equation}
where $n_{k}$ is the $k$-th component of the unit vector
$n=(n_{1},n_{2},n_{3})^{T}$ on Bloch sphere. The zero-discord states
are the ones which can be written in the form
\begin{equation}
  \rho=\sum_{k}p_{k}\Pi_{k}^{A}\otimes\rho_{k}^{B}
  \label{eq:zero-discord-form}
\end{equation}
with $p_{k}=\mathrm{Tr}(\Pi_{k}^{A}\rho)$ and
$\rho_{k}^{B}=\mathrm{Tr}_{A}(\Pi_{k}^{A}\rho)/p_{k}$, see
Ref.~\cite{Ollivier2001}. With Eq.~(\ref{eq:measurent_basis}), $p_{k}$
and $\rho_{k}^{B}$ can be obtained via

\begin{eqnarray}
  p_{1}\rho_{1}^{B} & = & \frac{1}{4}\sum_{i,j=0}^{3}(\alpha_{i}\tau_{ij})\sigma_{j}^{B},\nonumber \\
  p_{2}\rho_{2}^{B} & = & \frac{1}{4}\sum_{i,j=0}^{3}(\beta_{i}\tau_{ij})\sigma_{j}^{B}.
  \label{eq:pk_rhok}
\end{eqnarray}
Substituting Eqs.~(\ref{eq:measurent_basis}) and (\ref{eq:pk_rhok})
into the above form of the zero-discord
state~(\ref{eq:zero-discord-form}), we get
\begin{equation}
  \sum_{ij=0}^{3}\tau_{ij}\sigma_{i}^{A}\otimes\sigma_{j}^{B}=\frac{1}{2}\sum_{ijk=0}^{3}(\alpha_{i}\alpha_{k}+\beta_{i}\beta_{k})\tau_{kj}\sigma_{i}^{A}\otimes\sigma_{j}^{B}.
\end{equation}
Because $\{\sigma_{i}^{A}\otimes\sigma_{j}^{B}\}$ is a set of
orthogonal basis of operators on
$\mathcal{H}^{A}\otimes\mathcal{H}^{B}$, we get

\begin{equation}
  \tau=\frac{1}{2}(\alpha\alpha^{T}+\beta\beta^{T})\tau.
  \label{eq:prior_cond_zero_QD}
\end{equation}
where $\alpha=(\alpha_{0},\alpha_{1},\alpha_{2},\alpha_{3})^{T}$ and
$\beta=(\beta_{0},\beta_{1},\beta_{2},\beta_{3})^{T}$. The matrix
$\tau$ can be further decomposed into

\begin{equation}
  \tau=\left[\begin{array}{cc}
      1 & b^{T}\\
      a & R\end{array}\right],
  \label{eq:decomposition_tau}
\end{equation}
where $a$, $b$ are column vectors and $R$ are $3\times3$ matrix.
Combining Eqs.~(\ref{eq:alpha_beta}), (\ref{eq:prior_cond_zero_QD})
and (\ref{eq:decomposition_tau}), we obtain the condition of
zero-discord states with the Fano-Bloch representation as follows

\begin{eqnarray}
  nn^{T}a & = & a,\label{eq:zero-qd-cond1}\\
  nn^{T}R & = & R,\label{eq:zero-qd-cond2}
\end{eqnarray}
where $n$ is a unit column vector. The existence of such $n$
satisfying the above equations is a sufficient and necessary condition
for the zero-discord.

\section{a general expression of the quantum conditional
  entropy\label{sec:appd-general-expression-QCE}}

In the following, we give a general expression of the quantum
conditional entropy in the Fano-Bloch representation. For the states
(\ref{eq:Bloch_rep}) and the von Neumann measurement
(\ref{eq:measurent_basis}), the remaining state of the system $B$ with
measurement result $k=1,2$ of $A$ and the corresponding probability
are given by Eq.~(\ref{eq:pk_rhok}).

To obtain the eigenvalues of $\rho_{k}^{B}$, we first get the
eigenvalues of $p_{k}\rho_{k}^{B}$. From Eq.~(\ref{eq:pk_rhok}), we
have

\begin{align}
  \mathrm{eig}(p_{1}\rho_{1}^{B}) & =\frac{1}{4}(\sum_{i=0}^{3}\alpha_{i}\tau_{i0})\pm\frac{1}{4}\sqrt{\sum_{j=1}^{3}\left(\sum_{i=0}^{3}\alpha_{i}\tau_{ij}\right)^{2}},\\
  \mathrm{eig}(p_{2}\rho_{2}^{B}) &
  =\frac{1}{4}(\sum_{i=0}^{3}\beta_{i}\tau_{i0})\pm\frac{1}{4}\sqrt{\sum_{j=1}^{3}\left(\sum_{i=0}^{3}\beta_{i}\tau_{ij}\right)^{2}}.
\end{align}
With the decomposition form (\ref{eq:decomposition_tau}), we obtain
the following results

\begin{eqnarray}
  \mathrm{eig}(\rho_{1}) & = & \frac{1}{2}\left(1\pm\frac{\left|b+R^{T}n\right|}{1+a^{T}n}\right),\\
  \mathrm{eig}(\rho_{2}) & = & \frac{1}{2}\left(1\pm\frac{\left|b-R^{T}n\right|}{1-a^{T}n}\right),\\
  p_{1} & = & \frac{1}{2}(1+a^{T}n),\\
  p_{2} & = & \frac{1}{2}(1-a^{T}n),
\end{eqnarray}
where $|X|^{2}=\mathrm{Tr}(XX^{\dagger})$ is the Hilbert-Schmidt
norm. By introducing a new set of parameter as follows
\begin{eqnarray}
  f(n) & = & a^{T}n,\\
  g_{\pm}(n) & = & \left|b\pm R^{T}n\right|,
\end{eqnarray}
we get the quantum conditional entropy \cite{Ollivier2001} as follows

\begin{align}
  \sum_{k}p_{k}S(\rho_{k}^{B}) &
  =\frac{1+f}{2}h\left(\frac{g_{+}}{1+f}\right)+\frac{1-f}{2}h\left(\frac{g_{-}}{1-f}\right)
  \label{eq:general_QCE}
\end{align}
with
$h(x):=-\frac{1+x}{2}\log_{2}\frac{1+x}{2}-\frac{1-x}{2}\log_{2}\frac{1-x}{2}$.
Then the classical correlation $\mathcal{C}$ and the quantum discord
$\mathcal{D}$ can be obtained via
\begin{eqnarray}
  \mathcal{C} & = & S(\rho^{B})-\min_{|n|=1}\sum_{k}p_{k}S(\rho_{k}^{B}),\\
  \mathcal{D} & = & S(\rho^{A})-S(\rho)+\min_{|n|=1}\sum_{k}p_{k}S(\rho_{k}^{B}).
\end{eqnarray}

\bibliographystyle{apsrev} \bibliography{Discord}

\end{document}